\documentclass[a4paper]{article}

\usepackage[english]{babel}
\usepackage[utf8x]{inputenc}
\usepackage[T1]{fontenc}
\usepackage{amsmath,amsxtra,amssymb,amsfonts,amsthm,amscd,array}
\usepackage{graphicx,latexsym,wrapfig}
\DeclareMathOperator*{\argmax}{arg\,max}

\usepackage{color}
\usepackage{algorithm}
\usepackage{algorithmic}
\usepackage{mathtools}
\usepackage{float}
\usepackage{booktabs}
\usepackage{fancyhdr}
\usepackage{lscape}
\usepackage{dsfont}
\usepackage[colorlinks=true, allcolors=blue]{hyperref}
\usepackage[authoryear,sort]{natbib}
\usepackage{appendix}
\usepackage{tikz}
\usetikzlibrary{matrix,chains,positioning,decorations.pathreplacing,arrows}

\usepackage{xparse}
\usepackage{breakcites}

\newtheorem{theorem}{Theorem}
\newtheorem*{remark}{Remark}

\usepackage[a4paper,top=3cm,bottom=2cm,left=3cm,right=3cm,marginparwidth=1.75cm]{geometry}

\theoremstyle{definition}

\newtheorem{proposition}{Proposition}[section]
\newtheorem{lemma}{Lemma}[section]

\DeclarePairedDelimiterX{\expectarg}[1]{[}{]}{%
  \ifnum\currentgrouptype=16 \else\begingroup\fi
  \activatebar#1
  \ifnum\currentgrouptype=16 \else\endgroup\fi
}

\NewDocumentCommand\ee{moo} 
        {%
            \mathbb{E}\IfValueT{#3}{^{#3}}
            \left[#1
                \IfValueT{#2}{\middle|{#2}}
            \right]%
        }

\newcommand{\RR}{\mathbb{R}} \newcommand{\NN}{\mathbb{N}} 
\newcommand{\PP}{\mathbb{P}} 

\newcommand{\EE}{\mathbb{E}}

\title{Solving optimal stopping problems with Deep Q-learning}
\author{John Ery\footnote{RiskLab, Department of Mathematics, ETH Zurich, john.ery@math.ethz.ch} 
\and Loris Michel
\footnote{Seminar f\"ur Statistik, Department of Mathematics, ETH Zurich, michel@stat.math.ethz.ch}
}

\begin{document}
\maketitle

\begin{abstract}
We propose a reinforcement learning (RL) approach to model optimal exercise strategies for option-type products. 
We pursue the RL avenue in order to learn the optimal action-value function of the underlying stopping problem. In addition to retrieving the optimal Q-function at any time step, one can also price the contract at inception. We first discuss the standard setting with one exercise right, and later extend this framework to the case of multiple stopping opportunities in the presence of constraints. 
We propose to approximate the Q-function with a deep neural network, which does not require the specification of basis functions as in the least-squares Monte Carlo framework 
and is scalable to higher dimensions. We derive a lower bound on the option price obtained from the trained neural network and an upper bound from the dual formulation of the stopping problem, which can also be expressed in terms of the Q-function.
Our methodology is illustrated with examples covering the pricing of swing options. 
\end{abstract}

\section{Introduction}\label{sec:intro}
Reinforcement learning (RL) in its most general form deals with agents living in some environment and aiming at maximizing a given reward function. Alongside supervised and unsupervised learning, it is often considered as the third family of models in the machine learning literature. It encompasses a wide class of algorithms that have gained popularity in the context of building intelligent machines that can outperform masters in ancestral board games such as Go or chess, see e.g. \cite{silver16}; \cite{silver17}. These models are very skilled when it comes to learning the rules of a certain game, starting from little or no prior knowledge at all, and progressively developing winning strategies.
Recent research, see e.g. \cite{deepmind}, \cite{doubleQ}, \cite{duelingQ}, has considered integrating deep learning techniques in the framework of reinforcement learning in order to model complex unstructured environments. Deep reinforcement learning can hence leverage the ability of deep neural networks to uncover hidden structure from very complex functionals and the power of reinforcement techniques to take complex actions.

Optimal stopping problems from mathematical finance 
naturally fit into the reinforcement learning framework. Our work is motivated by the pricing of swing options which appear in energy markets (oil, natural gas, electricity) to hedge against futures price fluctuations, see e.g. \cite{meinshausen}, \cite{bender15}, and more recently \cite{daluiso20}.
Intuitively, when behaving optimally, investors holding these options are trying to maximize their reward by following some optimal sequence of decisions, which in the case of swing options consists in purchasing a certain amount of electricity or natural gas at multiple exercise times.

The stopping problems we will consider belong to the category of Markov decision processes (MDP). We refer the reader to \cite{puterman} or \cite{bertsekas} for good textbook references on this topic.  When the size of the MDP becomes large or when the MDP is not fully known (model-free learning), alternatives to standard dynamic programming techniques must be sought. Reinforcement learning can efficiently tackle these issues and can be transposed to our problem of determining optimal stopping strategies.

Previous work exists on the connections between optimal stopping problems in mathematical finance and reinforcement learning. For example, the common problem of learning optimal exercise policies for American options has been tackled in \cite{li} using reinforcement learning techniques. They implement two algorithms, namely least-squares policy iteration (LSPI), see \cite{lagoudakis}, and fitted Q-iteration (FQI), see \cite{vanroy}, and compare their performance to a benchmark provided by the least-squares Monte Carlo (LSMC) approach of see \cite{longstaff}. It is shown empirically that strategies uncovered by both these algorithms provide larger payoffs than LSMC. 
\cite{kohler08} model the Snell envelope of the underlying optimal stopping problem with a neural network. More recently, \cite{dos} derive optimal stopping rules from Monte Carlo samples at each time step using deep neural networks.
An alternative approach developed in \cite{becker20} considers the approximation of the continuation values using deep neural networks. This method also produces a dynamic hedging strategy based on the approximated continuation values. A similar approach with different activation functions is presented in \cite{lapeyre} alongside a convergence result for the pricing algorithm, whereas the method employed in \cite{chen20} is based on BSDE's.

Our work aims at casting the optimal stopping decision into a unifying reinforcement learning framework through the modeling of the action-value function of the problem. One can then leverage reinforcement learning algorithms involving neural networks that learn the optimal action-value function at any time step. We illustrate this methodology by presenting examples from mathematical finance. In particular, we will focus on high-dimensional swing options where the action taking is more complex, and where deep neural networks are particularly powerful due to their approximation capabilities.

The remainder of the paper is structured as follows. In Section \ref{sec:methodo}, we introduce the necessary mathematical tools from reinforcement leaning, present an estimation approach of the Q-function using neural networks, and discuss the derivation of the lower and upper bounds on the option price. In Section \ref{sec:multiple} we explain the multiple stopping problem with waiting period constraint between two consecutive exercise times, again with the derivation of a lower bound and an upper bound on the option price at inception.
We display numerical results for swing options in Section \ref{sec:results}, and conclude in Section \ref{sec:conclu}.

\section{Theory and methodology}\label{sec:methodo}
In this section we present the mathematical building blocks and the reinforcement learning machinery, leading to the formulation of the stopping problems under consideration. 
\subsection{Markov decision processes and action-value function}\label{sec:MDP}
As discussed in the introduction, the problems we will consider in the sequel can be embedded into the framework of the well-studied Markov decision processes (MDPs), see \cite{sutton}.
A Markov decision process is defined as a tuple $\left(\mathcal{S},\mathcal{A},p,R,\gamma\right),$ where 
\begin{itemize}
\item $\mathcal{S}$ is the set of states;
\item $\mathcal{A}$ is the set of actions the agent can take;
\item $p$ is the transition probability kernel, where $p(\cdot|s,a)$ is the probability of future states given that the current state is $s$ and that action $a$ is taken;
\item $R$ is a reward function, where $R(s,a)$ denotes the reward obtained when moving from state $s$ under action $a$ (note here that different definitions exist in the literature);
\item $\gamma\in(0,1)$ is a discount factor which expresses preference towards short-term rewards (in the present work $\gamma=1$ as we consider already discounted rewards).
\end{itemize}
A policy $\pi$ is then a rule for selecting actions based on the last visited state. More specifically, $\pi(s,a)$ denotes the probability of taking action $a$ in state $s$ under policy $\pi.$
The conventional task is to maximize the total (discounted) expected reward over policies, and can be expressed as $\EE_{\pi}\left[\sum_{t=0}^\infty \gamma^t R_t\right].$ A policy which maximizes this quantity is called an optimal policy. 
Given a starting state $s_0,$ an initial action $a_0,$ one can define the \emph{action-value function}, also called \emph{Q-function}:
\begin{equation}\label{eq:Qfunc}
Q^{\pi}(s,a)=\EE_{\pi}\left[\sum_{t=0}^\infty \gamma^t R_t\Big \lvert s_0=s,\mbox{ }a_0=a\right],
\end{equation}
where $R_t=R(s_t,a_t),$
for a sequence of state-action pairs $(s_t,a_t)_{t\geq 0}\sim \pi.$
The optimal policy $\pi^*$ satisfies
\begin{equation}
Q^{*}(s,a)=\sup_{\pi}Q^{\pi}(s,a),
\label{eq:Qopt}
\end{equation} 
where we write $Q^{*}$ for $Q^{\pi^*}$. 
In other words, the optimal Q-function measures how "good" or "rewarding" it is to choose action $a$ while in state $s,$ by following optimal decisions.
We will consider problems with finite time horizon $T>0,$ and we accordingly set $R_t=0$ for all $t>T.$

\subsection{Single stopping problems as Markov decision processes}
We consider the same stopping problem as in \cite{dos} and \cite{becker20}, namely an American-style option defined on a finite time grid $0 = t_0 < t_1 < \ldots < t_N = T$. The discounted payoff process $\left(G_n\right)_{n=0}^N$ is assumed to be square-integrable and takes the form $G_n=g\left(n, X_{t_n}\right)$ for a measurable function $g: \{0, 1, \ldots, N\} \times \mathbb{R}^{d} \rightarrow \mathbb{R}$ and a $d$-dimensional $\mathbb{F}$-Markovian process $\left(X_{t_n}\right)_{n=0}^N$ defined on a filtered probability space $\left(\Omega, \mathcal{F}, \mathbb{F}=\left(\mathcal{F}_n\right)_{n=0}^N, \mathbb{P}\right)$.
Let $E\subset\RR^d$ denote the space in which the underlying process lives.
We assume that $X_0$ is deterministic and that $\PP$ is the risk-neutral probability measure.
The value of the option at time $0$ is given by 
\begin{equation}V_0=\sup_{\tau \in \mathcal{T}} \EE\left[g(\tau, X_{\tau})\right],
\label{eq:stop}
\end{equation}
where $\mathcal{T}$ denotes all stopping times $\tau: \Omega \rightarrow \{t_0,t_1,\ldots,t_N\}$. This problem is essentially a Markov decision process with state space $\mathcal{S} = \{0,1,\ldots,N\}\times\mathbb{R}^{d}\times  \{0,1\}$, action space $\mathcal{A} = \{0,1\}$ (where we follow the convention $a=0$ for continuing and $a=1$ for stopping), reward function\footnote{When exercizing (taking action $a=1$), we implicitly move to the absorbing state, i.e. the last component of the state space becomes 1.}  \begin{equation*}
R\left(\left(n, X_{t_n}\right), a\right) = \begin{cases} g(n, X_{t_n}), & \text{if } a = 1,\\ 0, &\text{if } a = 0,\end{cases}
\end{equation*}
for $n=0,\ldots,N,$ and transition kernel $p$ driven by the dynamics of the $\mathbb{F}$-Markovian process $(X_{t_n})_{n=0}^N$.
The state space includes time, the $d$-dimensional Markovian process and an additional (absorbing) state which at each time step captures the event of exercise or no exercise. More precisely, we jump to this absorbing state when we have exercised. In the multiple stopping case which we discuss in Section \ref{sec:multiple}, we jump to this absorbing state once we have used the last exercise right. In both single and multiple stopping frameworks, once this absorbing state has been reached at a random time 
$\tau: \Omega \rightarrow \{t_0,t_1,\ldots,t_N\}$, we set all rewards and Q-values to 0 for $t>\tau.$ 
The associated \emph{Snell envelope} process $\left(Z_n\right)_{n=0}^N$ of the stopping problem in (\ref{eq:stop}) is defined recursively by 
\begin{equation}
Z_n =\begin{cases} 
g\left(N, X_{t_N}\right), &\text{ if } n = N,\\ \max\left\{g\left(n, X_{t_n}\right), \EE\left[Z_{n+1} \mid \mathcal{F}_n\right]\right\}, &\text{ if } 0 \leq n \leq N-1.
\end{cases}\label{eq:snell}
\end{equation}
It is well known that the Snell envelope provides an optimal stopping time solving (\ref{eq:stop}) as stated in the following result.\footnote{Note that in particular $Z_{0}=V_{0}$.} A standard proof for the latter can be found in \cite{shreve}.

\begin{proposition}\label{optimal stop}
The stopping time $\tau^{*}$ defined by
\begin{equation*}
\tau^{*} = \inf\{n: Z_n = g\left(n, X_{t_n}\right)\}
\end{equation*}for the Snell envelope $\left(Z_n\right)_{n=0}^N$ given in (\ref{eq:snell}), is optimal for the problem (\ref{eq:stop}).
\end{proposition}

\noindent
Various modeling approaches have been proposed to estimate the option value in (\ref{eq:stop}). \cite{kohler08} propose to model directly the Snell envelope, \cite{dos} take the approach of modeling the optimal stopping times. More recently, \cite{becker20} model the continuation values of the stopping problem. In this work, we rather propose to model the optimal action-value function of the problem $Q^{*}\left(\left(n, X_{t_n}\right), a\right)$ for all $n=0,\ldots,N,$ and $a \in \{0,1\}$ (where $a$ represents the stopping decision) given by 
\begin{equation}
Q^{*}\left(\left(n, X_{t_n}\right), a\right) =\begin{cases}  g\left(n,X_{t_n}\right), &\text{ if } a = 1,\\ \mathbb{E}\left[Z_{n+1} \mid \mathcal{F}_n\right], &\text{ if } a = 0.
\end{cases}\label{eq:Qstar}
\end{equation}
\sloppy
According to Proposition \ref{optimal stop}, through the knowledge of the optimal action-value function $Q^{*}\left(\left(n, X_{t_n}\right), a\right),$ we can recover the optimal stopping time $\tau^{*}$. Indeed, it turns out that the optimal decision functions $f_{0}, \ldots, f_N$ in \cite{dos} can be expressed in the action-value function framework through
\begin{equation*}
f_{n}\left(X_{t_n}\right) = \mathds{1}\left\{\argmax_{a \in \{0,1\}}Q^{*}\left(\left(n, X_{t_n}\right), a\right) = 1\right\}, \mbox{ } \forall  \mbox{ } n=0,\ldots, N,
\end{equation*}
where $\mathds{1}\{\cdot\}$ denotes the indicator function. Moreover, one can express the Snell envelope (estimated in \cite{kohler08}) as $Z_n = \max\left\{ Q^{*}\left(\left(n, X_{t_n}\right), 0\right), Q^{*}\left(\left(n, X_{t_n}\right), 1\right)\right\},$ and the continuation value modeled in \cite{becker20} can be reformulated in our setting as $C_n = Q^{*}\left(\left(n, X_{t_n}\right), 0\right)$.
As a by-product, one can price financial products such as swing options 
by considering
$\max\{Q^{*}(0,X_0,1),Q^{*}(0,X_0,0)\}.$ 

In this perspective, our modeling approach is very similar to previous studies but differs in the reinforcement learning machinery employed. Indeed, modeling the action-value function and optimizing it is a common and natural approach known under the name of \emph{Q-learning} in the reinforcement learning literature. We introduce it in the next section. 
\subsection{Q-learning as estimation method}
In contrast to policy or value iteration, Q-learning methods, see e.g. \cite{watkins89} and \cite{watkins92}, estimate directly the optimal action-value function. They are model-free and can learn optimal strategies with no prior knowledge of the state transitions and the rewards. In this paradigm, an agent interacts with the environment (exploration step) and learns from past actions (exploitation step) to derive the optimal strategy.

One way to model the action-value function is by using deep neural networks. This approach is referred to under the name deep Q-learning in the reinforcement learning literature. In this setup, the optimal action-value function $Q^{*}$ is modeled with a neural network $Q\left(s,a;\theta\right)$ often called deep Q-network (DQN), where $\theta$ is a vector of parameters corresponding to the network architecture. 
However, reinforcement learning can be highly unstable or even potentially diverge due to the introduction of neural networks in the approximation the Q-function.
To tackle these issues, a variant to the original Q-learning method has been developed in \cite{mnih15}. It relies on two main concepts.
The first is called experience replay and allows to remove correlations in the sequence of observations.
In practice this is done by generating a large sample of experiences which we denote as vectors $e_t=(s_t,a_t,r_t,s_{t+1})$ at each time $t,$ and that we store in a dataset $D.$ We note that once we have reached the absorbing state, we start a new episode or sequence of observations by resetting the MDP to the initial state $s_0$.
Furthermore, we allow the agent to explore new unseen states according to a so-called $\varepsilon$-greedy strategy, see \cite{sutton}, meaning that with probability $\varepsilon$ we take a random action and with probability $(1-\varepsilon)$ we take the action maximizing the Q-value. Typically one reduces the value of $\varepsilon$ according to a linear schedule as the training iterations increase.

During the training phase, we then perform updates to the Q-values by sampling mini-batches uniformly at random from this dataset $(s,a,r,s')\sim \mathcal{U}(D)$ and minimizing over $\theta$
the following loss function
\begin{equation}\label{eq:lossDQN}
L(\theta)=
\EE_{(s,a,r,s')\sim \mathcal{U}(D)}\left[
\left(R(s,a)+\gamma \max_{a'} Q\left(s',a';\theta\right)-Q\left(s,a;\theta\right)\right)^2\right].
\end{equation}
However there might still be some correlations between the Q-values $Q(s,a;\theta)$ and the so-called target values $R(s,a)+\gamma \max_{a'} Q\left(s',a';\theta\right).$ The second improvement brought forward in \cite{mnih15} consists in
updating the network parameters for the target values only with a regular frequency and not after each iteration. This is called parameter freezing and translates into minimizing over $\theta$ the modified loss function
\begin{equation}\label{eq:lossDQNfreezing}
L(\theta)=
\EE_{(s,a,r,s')\sim \mathcal{U}(D)}\left[
\left(R(s,a)+\gamma \max_{a'} Q\left(s',a';\theta^*\right)-Q\left(s,a;\theta\right)\right)^2\right],
\end{equation}
where the target network parameters $\theta^*$ are only updated with the DQN parameters $\theta$ every $T^*>0$ steps, and are held constant between individual updates.

An alternative network specification would be to take only the state as input $Q\left(s;\theta\right)$ and update the Q-values for each action, see the implementation in \cite{deepmind}. 
Network architectures such as double deep Q-networks, see \cite{doubleQ}, dueling deep Q-networks, see \cite{duelingQ}, and combinations thereof, see \cite{rainbow} have been developed to improve the training performance even further. However the implementation of these algorithms is out of the scope of our presentation.

\subsection{Inference and confidence intervals}
In the same spirit as \cite{dos} and \cite{becker20}, we compute lower and upper bounds on the option price in (\ref{eq:stop}), the confidence interval resulting from the central limit theorem, as well as a point estimate for the optimal value $V_0.$ In the sequel, for ease of notation, we will use $X_{t_n}=X_n,$ for $n=0,\ldots,N.$

\subsubsection{Lower bound}\label{sec:lb}
We store the parameters learned through the training of the deep neural network on an experience replay dataset with simulations $\left(X_n^k\right)_{n=0}^N$ for $k=1,\ldots,K.$ 
We denote as $\hat{\theta}\in\Theta$ the vector of network parameters where $\Theta\subset \RR^q,$ 
$q>0$ denotes the dimension of the parameter space and $Q\left(s,a;\hat{\theta}\right)$ corresponds to the calibrated network. We then generate new simulations of the state space process $\left(X_n^k\right)_{n=0}^N$, independent from those used for training, for $k=K+1,\ldots,K+K_L.$ The independence is necessary to achieve unbiasedness of the estimates. The Monte Carlo average 
\begin{equation*}
\hat{L}=\frac{1}{K_L}\sum_{k=K+1}^{K+K_L}g\left(\tau_{k},X^{k}_{\tau_{k}}\right)
\end{equation*}
where 
$\tau_k = \inf \left\{0 \leq n \leq N: Q\left(\left(n,X_n^{k}\right),1;\hat{\theta}\right) > Q\left(\left(n,X_n^{k}\right),0;\hat{\theta}\right) \right\}$ yields a lower bound for the optimal value $V_0.$ Since the optimal strategies are not unique, we follow the convention of taking the largest optimal stopping rule which yields a strict inequality.

\subsubsection{Upper bound}\label{sec:upper}
The derivation of the upper bound is based on the Doob-Meyer decomposition of the supermartingale given by the Snell envelope, see \cite{shreve}.
The Snell envelope $(Z_n)_{n=0}^N$ of the discounted payoff process $(G_n)_{n=0}^N$ can be decomposed as
\begin{equation*}
Z_n=Z_0+M_n^Z-A_n^Z,    
\end{equation*}
where $M^Z$ is the $(\mathcal{F}_n)$-martingale given by 
\begin{equation*}
M_0^Z=0 \textnormal{  and  } M_n^Z-M_{n-1}^Z=Z_n-\ee{Z_n\mid \mathcal{F}_{n-1}},\quad n=1,\ldots,N,   
\end{equation*}
and $A^Z$ is the non-decreasing $(\mathcal{F}_n)$-predictable process given by 
\begin{equation*}
A_0^Z=0 \textnormal{  and  } A_n^Z-A_{n-1}^Z=Z_{n-1}-\ee{Z_n\mid \mathcal{F}_{n-1}},\quad n=1,\ldots,N.  
\end{equation*}
From Proposition 7 in \cite{dos}, given a sequence $(\varepsilon_n)_{n=0}^N$ of integrable random variables in $(\Omega,\mathcal{F},\mathbb{P})$ such that $\mathbb{E}[\varepsilon_n\mid\mathcal{F}_n]=0$ for all $n=0,\ldots,N,$
one has
\begin{equation*}
V_0\leq \ee{\max_{0\leq n \leq N} \left(g(n,X_n)-M_n-\varepsilon_n\right)},
\end{equation*}
for every $(\mathcal{F}_n)$-martingale $\left(M_n\right)_{n=0}^N$ starting from 0.

This upper bound is tight if $M=M^Z$ and $\varepsilon\equiv 0.$ We can then use the optimal action-value function learned via the deep neural network to construct a martingale close to $M^Z.$
We now adapt the approach presented in \cite{dos} to the expression of the martingale component of the Snell envelope. Indeed, the martingale differences $\Delta M_n$ from Subsection 3.2 in \cite{dos} can be written in terms of the optimal action-value function:
\begin{equation*}
\Delta M_n=M_n-M_{n-1}=Q^{*}((n,X_n),a)-Q^{*}((n-1,X_{n-1}),0),
\end{equation*}
since the continuation value at time $n-1$ is given by evaluating the optimal action-value function at action $a=0$ (continuing).
Given the definition of the optimal action-value function at (\ref{eq:Qstar}), one can rewrite the martingale differences as
\begin{equation}
\Delta M_n=g(n,X_n)\mathds{1}\{a=1\}+\mathbb{E}\left[Z_{n+1} \mid \mathcal{F}_n\right]\mathds{1}\{a=0\}-
\mathbb{E}\left[Z_{n} \mid \mathcal{F}_{n-1}\right].
\end{equation}
The empirical counterparts are given
by generating realizations $M_n^k$ of $M_n+\varepsilon_n$
based on a sample of $K_U$ simulations $\left(X_n^k\right)_{n=0}^N,$ for $k=K+K_L+1,\ldots,K+K_L+K_U.$
Again, we simulate realizations of the state space process independently from the simulations used for training. This gives us the following empirical differences:
\begin{equation*}
\Delta M_n^k= g\left(n,X_n^k\right)\mathds{1}\left\{a_n^k=1\right\}+\hat{\mathbb{E}}\left[Z_{n+1}^k \mid \mathcal{F}_{n}\right]\mathds{1}\left\{a_n^k=0\right\}-\hat{\mathbb{E}}\left[Z_n^k \mid \mathcal{F}_{n-1}\right],   
\end{equation*}
where $a_n^k$ is the chosen action at time $n$ for simulation path $k,$ and $\hat{\EE}\left[Z_{n+1}^k \mid \mathcal{F}_{n}\right]$ are the Monte Carlo averages approximating the continuation values 
for $n=0,\ldots,N-1$ and $k=K+K_L+1,\ldots,K+K_L+K_U.$

The continuation values appearing in the martingale increments are obtained through nested simulation, see the remark below:
\begin{equation*}
\hat{\mathbb{E}}\left[Z_{n+1}^k \mid \mathcal{F}_{n}\right]=
\frac{1}{J}\sum_{j=1}^{J}g\left(\tau_{n+1}^{k,j},\widetilde{X}_{\tau_{n+1}^{k,j}}^{k,j}\right),
\end{equation*}
where $J$ is the number of simulations in the inner step, and where, given each $X_n^k,$ we simulate (conditional) continuation paths $\widetilde{X}_{n+1}^{k,j},\ldots,\widetilde{X}_N^{k,j},$ $j=1,\ldots,J,$ that are conditionally independent of each other and of $X_{n+1}^k, \ldots,X_N^k,$ and $\tau_{n+1}^{k,j}$ is the value of $\tau_{n+1}^{\theta}$ along the path $\widetilde{X}_{n+1}^{k,j},\ldots,\widetilde{X}_N^{k,j}.$
\begin{remark}
It is not guaranteed than $\ee{\Delta M_n}[\mathcal{F}_{n-1}]=0$ for the Q-function learned via the neural network. To tackle this issue, we implement nested simulations as in \cite{dos} and \cite{becker20} to estimate the continuation values. This gives unbiased estimates of $M_n,$ which is crucial to obtain a valid upper bound. Moreover, the variance of the estimates decreases with the number of inner simulations, at the expense of increased computational time.
\end{remark}

\noindent 
Finally we can derive an unbiased estimate for the upper bound of the optimal value $V_0:$
\begin{equation*}
\hat{U}=\frac{1}{K_U}\sum_{k=K+K_L+1}^{K+K_L+K_U}\max_{0\leq n\leq N}\left(g\left(n,X_n^k\right)-M_n^k\right), \end{equation*}
with $M_n^k=\sum_{m=1}^n\Delta M_m^k.$
\subsubsection{Point estimate and confidence interval}\label{sec:CI}
The average between the lower and the upper bound for the point estimate of $V_0$ is considered in \cite{dos} and \cite{becker20}:
\begin{equation*}
\frac{\hat{L}+\hat{U}}{2}.
\end{equation*}
Assuming the discounted payoff process is square-integrable for all $n=0,\ldots,N,$ we also obtain that the upper bound $\max_{0\leq n \leq N}\left(g(n,X_n)-M_n-\varepsilon_n\right)$ is square-integrable. 
Let $z_{\alpha/2}$ denote the $(1-\alpha/2)$-quantile of a standard normal distribution. Defining the empirical standard deviations for the lower and upper bounds as
\begin{equation*}
\hat{\sigma}_L=\sqrt{\frac{1}{K_L-1}\sum_{k=K+1}^{K+K_L}\left(X^{k}_{\tau_{k}}-\hat{L}\right)^2},    
\end{equation*}
and
\begin{equation*}
\hat{\sigma}_U=\sqrt{\frac{1}{K_U-1}\sum_{k=K+K_L+1}^{K+K_L+K_U}\left(\max_{0\leq n\leq N}\left(g\left(n,X_n^k\right)-M_n^k\right)-\hat{U}\right)^2},
\end{equation*}
respectively, one can leverage the central limit theorem to build the asymptotic two-sided $(1-\alpha)$-confidence interval for the true optimal value $V_0:$
\begin{equation}
\left[\hat{L}-z_{\alpha/2}\frac{\hat{\sigma}_L}{\sqrt{K_L}},\hat{U}+z_{\alpha/2}\frac{\hat{\sigma}_U}{\sqrt{K_U}}\right].
\end{equation}
We have presented in this section the unifying properties of Q-learning compared to other approaches used to study optimal stopping problems. On one hand we do not require any iterative procedure and do not have to solve a potentially complicated optimization problem at each time step. Indeed the calibrated deep neural network solves the optimal stopping problem on the whole time interval.
On the other hand, we are able to accommodate any finite number of possible actions. Looking back at the direct approach of \cite{dos} to model optimal stopping policies, the parametric form of the stopping times would explode if we allow for more than two possible actions.

\section{Multiple stopping with constraints}\label{sec:multiple}
In this section we extend the previous problem to the more general framework of multiple-exercise options.
Examples from this family include swing options, which are common in the electricity market. The holder of such an option is entitled to exercise a certain right, e.g. the delivery of a certain amount of energy, several times, until the maturity of the contract.
The number of exercise rights and constraints on how they can be used are specified at inception. Typical constraints are a waiting period, i.e. a minimal waiting time between two exercise rights, and a volume constraint, which specifies how many units of the underlying asset can be purchased at each time.

Monte Carlo valuation of such products has been studied in \cite{meinshausen}, producing lower and upper bounds for the price. Building on the dual formulation for option pricing, alternative methods additionally accounting for waiting time constraints have been considered in \cite{bender11}, and for both volume and waiting time constraints in \cite{bender15}. In all cases, the multiple stopping problem is decomposed into several single stopping problems using the so-called reduction principle. The dual formulation in \cite{meinshausen} expresses the marginal excess value due to each additional exercise right as an infimum of an expectation over a certain space of martingales and a set of stopping times. A version of the dual problem in discrete time relying solely on martingales is presented in \cite{schoenmakers}, and a dual for the continuous time problem with a non-trivial waiting time constraint is derived in \cite{bender11}. In the latter case, the optimization is not only over a space of martingales, but also over adapted processes of bounded variation, which stem from the Doob-Meyer decomposition of the Snell envelope. The dual problem in the more general setting considering both volume and waiting time constraints is formulated in \cite{bender15}.

We now express the multiple stopping extension of the problem defined at (\ref{eq:stop}) for American-style options. Assume that the option holder has $n>0$ exercise rights over the lifetime of the contract. We consider the setting with no volume constraint and a waiting time $\delta>0$ which we assume to be a multiple of the time step resulting from the discretization of the interval $[0,T].$ The action space is still $\mathcal{A}=\{0,1\}.$ The state space now has an additional dimension corresponding to the number of remaining exercise opportunities.
As in standard stopping, we assume an absorbing state to which we jump once the $n$-th right has been exercised.

We note that due to the introduction of the waiting period, depending on the specification of $n,$ $T$ and $\delta,$ it may not be possible for the option holder to exercise all his rights before maturity, see the discussion in \cite{bender15}, where a "cemetery time" is defined.
If the specification of these parameters allows the exercise of all rights, and if we assume that $g\left(n,X_{t_n}\right)\geq 0$ for all $n=0,\ldots,N,$ then it will always be optimal to use all exercise rights.
The value of this option with $n>1$ exercise possibilities at time $0$ is given by 
\begin{equation}V_0^n=\sup_{\tau \in \mathcal{T}_\delta^n}\sum_{i=1}^n \mathbb{E}\left[g(\tau_i, X_{\tau_i})\right],
\label{eq:stopMult}
\end{equation}
where $\mathcal{T}_\delta^n$ is the set of $n$-tuples $\tau=(\tau_n,\tau_{n-1},\ldots,\tau_{1})$ of stopping times in $\{t_0,t_1,\ldots,t_N\}^n$ satisfying $\tau_i\geq \tau_{i+1}+\delta,$ for $i=1,\ldots,n-1.$

As in \cite{bender11}, one can combine the dynamic programming principle with the reduction principle to rewrite the primal optimization problem. We introduce the following functions defined in \cite{bender11} for $\nu=1,\ldots,n,$ and $k=N,\ldots,0:$
\begin{equation*}
q^{\nu}(k,x)=\ee{y^{\nu}\left(t_{k+1},X_{t_{k+1}}\right)}[X_{t_k}=x],
\end{equation*}
\begin{equation*}
q^{\nu}_{\delta}(k,x)=\ee{y^{\nu}\left(t_{k+\delta},X_{t_{k+\delta}}\right)}[X_{t_k}=x],
\end{equation*}
and we define the functions $y^\nu$ as 
\begin{equation*}
y^\nu(t_k,x)=\max\{g\left(t_k,x\right)+q_\delta^{\nu-1}\left(k,x\right),q^{\nu}\left(k,x\right)\}. \end{equation*}
We set $q^0_\delta(k,x)=0$ for all $x\in\RR^d$ and all $k\in\{0,\ldots,N\},$ and $g(t,X_t)=0$ for all $t>T.$ In the sequel, we denote as $y^{*,\nu}$ the Snell envelope for the problem with $\nu$ remaining exercise rights, for $\nu=1,\ldots,n.$
The reduction principle essentially states that the option with $n$ stopping times is as good as the single option paying the immediate cashflow plus the option with $(n-1)$ stopping times starting with a temporal delay of $\delta.$ This philosophy is also followed in \cite{meinshausen} by looking at the marginal extra payoff obtained with an additional exercise right.
The function $q^\nu$ corresponds to the continuation value in case of no exercise and the function $q^\nu_{\delta}$ to the continuation value in case of exercise, which requires a waiting period of $\delta.$

As shown in \cite{bender11}, one can derive the optimal policy from the continuation values.
Indeed, the optimal stopping times $\tau^{*,n}_\nu,$ for $\nu=1,\ldots,n,$ are given by
\begin{equation}
\tau^{*,n}_\nu=\inf\left\{k\geq \tau^{*,n}_{\nu+1}+\delta;
g\left(t_k,x\right)+q^{*,\nu-1}_{\delta}\left(k,x\right)\geq q^{*,\nu}\left(k,x\right)\right\},
\end{equation}
for starting value $\tau^{*,n}_{n+1}=-\delta$, which is a convention to make sure that the first exercise time is bounded from below by 0.
The optimal price is then 
\begin{equation*}
V_0^n=y^{*,n}\left(0,X_0\right),
\end{equation*}
and as in the single stopping framework, one can express the Snell envelope, the optimal stopping times and the continuation values in terms of the optimal Q-function $Q^*.$
Indeed, the continuation values can be expressed as
\begin{equation*}
q^{*,\nu}(k,X_{t_k})=Q^*\left(\left(t_k,X_{t_k}\right),\nu\right),
\end{equation*}
\begin{equation*}
q^{*,\nu}_{\delta}(k,X_{t_{k+\delta}})=Q^*\left(\left(t_{k+\delta},X_{t_{k+\delta}}\right),\nu\right),
\end{equation*}
the Snell envelope as
\begin{equation*}
y^{*,\nu}(t_k,X_{t_k})=\max\left\{g\left(t_k,X_{t_k}\right)+Q^*\left(\left(t_{k+\delta},X_{t_{k+\delta}}\right),\nu-1\right),Q^*\left(\left(t_k,X_{t_k}\right),\nu\right)
\right\}, 
\end{equation*}
and the optimal policy as
\begin{equation}\label{eq:stoppingMult}
\tau^{*,n}_\nu=\inf\left\{k\geq \tau^{*,n}_{\nu+1}+\delta;
g\left(t_k,X_{t_k}\right)+Q^*\left(\left(t_{k+\delta},X_{t_{k+\delta}}\right),\nu-1\right) \geq Q^*\left(\left(t_k,X_{t_k}\right),\nu\right)\right\}.
\end{equation}
To remain consistent with the notation introduced above for the functions $q^\nu,$ $q_{\delta}^\nu$ and $y^\nu,$ we denote by $Q^*\left(\left(t_k,X_{t_k}\right),\nu\right)$ the optimal Q-value in state $(t_k,X_{t_k},\nu),$ i.e. when there are $\nu$ remaining exercise rights.
Analogously to standard stopping with one exercise right, we can derive a lower bound from the primal problem and an upper bound from the dual problem. Moreover, we derive a confidence interval around the pointwise estimate based on Monte Carlo simulations.
\subsection{Lower bound}
As in Section \ref{sec:lb}, we denote by $Q\left(s,a;\hat{\theta}\right)$ the deep neural network calibrated through the training process using experience replay on a sample of simulated paths $\left(X_n^m\right)_{n=0}^N$ for $m=1,\ldots,M.$ 
We then generate a new set of $M_L$ simulations $\left(X_n^m\right)_{n=0}^N$, independent from the simulations used for training, for $m=M+1,\ldots,M+M_L.$ 
Then, using the learned stopping times 
\begin{equation*}
\tau_\nu^{m,n}=\inf\left\{k\geq \tau^{m,n}_{\nu+1}+\delta;
g\left(t_k,X_{t_k}^m\right)+Q\left(\left(t_{k+\delta},X_{t_{k+\delta}}^m\right),\nu-1;\hat{\theta}\right) \geq Q\left(\left(t_k,X_{t_k}^m\right),\nu;\hat{\theta}\right)\right\},
\end{equation*} 
for $\nu=1,\ldots,n,$ and with the convention $\tau^{m,n}_{n+1}=-\delta$ for all $m=M+1,\ldots,M+M_L,$
the Monte Carlo average
\begin{equation*}
\hat{L}^n=\frac{1}{M_L}\sum_{m=M+1}^{M+M_L} \sum_{\nu=1}^n g\left(\tau_\nu^{m,n},X_{\tau_\nu}^{m}\right)
\end{equation*}
yields a lower bound for the optimal value $V_0^n.$ In order to not overload the notation we consider $\tau_\nu=\tau_\nu^{m,n}$ in the subscript of the simulated state space above.
\subsection{Upper bound}
By exploiting the dual as in \cite{bender11}, one can also derive an upper bound on the optimal value $V_0^n.$ In order to do so, we consider the Doob decomposition of the supermartingales $y^{*,\nu}\left(t_k,X_{t_k}\right)$ given by
\begin{equation*}
y^{*,\nu}\left(t_k,X_{t_k}\right)=y^{*,\nu}\left(0,X_0\right)+M^{*,\nu}(k)-A^{*,\nu}(k),
\end{equation*}
where $M^{*,\nu}(k)$ is a $(\mathcal{F}_k)$-martingale with $M^{*,\nu}(0)=0$ and $A^{*,\nu}(k)$ is a non-decreasing  $(\mathcal{F}_k)$-predictable process with $A^{*,\nu}(0)=0$ for all $\nu=1,\ldots,n$ and $k=0,\ldots,N.$
The corresponding approximated terms using the learned Q-function lead to the following decomposition:
\begin{equation*}
y^{\nu}\left(t_k,X_{t_k}\right)=y^{\nu}\left(0,X_0\right)+M^{\nu}(k)-A^{\nu}(k), \end{equation*}
where $M^{\nu}$ are martingales with $M^\nu(0)=0,$ for $\nu=1,\ldots,n,$ and $A^{\nu}$ are integrable adapted processes in discrete time with $A^{\nu}(0)=0,$ for $\nu=1,\ldots,n.$ 

Moreover, one can write the increments of both the martingale and adapted components as:
\begin{equation*}
M^\nu(k)-M^\nu(k-1)=y^\nu\left(t_k,X_{t_k}\right)-\ee{y^\nu\left(t_k,X_{t_k}\right)}[\mathcal{F}_{k-1}],
\end{equation*} and
\begin{equation*}
A^\nu(k)-A^\nu(k-1)=y^\nu\left(t_{k-1},X_{t_{k-1}}\right)-\ee{y^\nu\left(t_k,X_{t_k}\right)}[\mathcal{F}_{k-1}].
\end{equation*}
Given the existence of the waiting period, one must also include the $\delta$-increment term
\begin{equation*}
A^\nu(k+\delta)-\ee{A^\nu(k+\delta)}[\mathcal{F}_k]
=M^\nu(k+\delta)-M^\nu(k)+\ee{y^\nu\left(t_{k+\delta},X_{t_{k+\delta}}\right)}[\mathcal{F}_k]-y^\nu\left(t_{k+\delta},X_{t_{k+\delta}}\right).
\end{equation*}
We note that for $\delta=1,$ since $A^{*,\nu}$ is a predictable process, this increment is equal to 0 for the optimal martingale $M^{*,\nu}$ and we retrieve the dual formulation in \cite{schoenmakers}.

As the dual formulation involves conditional expectations, we use nested simulation on a new set of $M_U$ independent simulations $\left(X_n^m\right)_{n=0}^N$ for $m=M+M_L+1,\ldots,M+M_L+M_U,$ with $M_U^{inner}$ inner simulations for each outer simulation as explained in Section \ref{sec:upper}, to approximate the one-step ahead continuation values
$\ee{y^\nu\left(t_k,X_{t_k}\right)}[\mathcal{F}_{k-1}]$ and the $\delta$-steps ahead continuation values 
$\ee{y^\nu\left(t_{k+\delta},X_{t_{k+\delta}}\right)}[\mathcal{F}_k].$ We denote the Monte Carlo estimators of these conditional expectations as
$\hat{\mathbb{E}}\left[y^\nu(t_k,X_{t_k})\mid \mathcal{F}_{k-1}\right]$
and
$\hat{\mathbb{E}}\left[y^\nu(t_{k+\delta},X_{t_{k+\delta}})\mid \mathcal{F}_k\right],$ respectively. We use these quantities to express the empirical counterparts of the adapted process increments for $m=M+M_L+1,\ldots,M+M_L+M_U:$
\begin{equation*}
A_m^\nu(k+\delta)-\hat{\mathbb{E}}\left[A^\nu(k+\delta)\mid \mathcal{F}_k\right]=M_m^\nu(k+\delta)-M_m^\nu(k)+\hat{\mathbb{E}}\left[y^\nu(t_{k+\delta},X_{t_{k+\delta}})\mid \mathcal{F}_k\right]-y_m^\nu\left(t_{k+\delta},X_{t_{k+\delta}}^m\right).
\end{equation*}
We can then rewrite the empirical counterparts of the Snell envelopes through the Q-function:
\begin{equation*}
y_m^\nu(t_k,X_{t_k}^m)=\max\left\{g\left(t_k,X_{t_k}^m\right)+Q\left(\left(t_{k+\delta},X_{t_{k+\delta}}^m\right),\nu-1;\hat{\theta}\right),Q\left(\left(t_k,X_{t_k}^m\right),\nu;\hat{\theta}\right)
\right\},
\end{equation*}
for $\nu=1,\ldots,n,$ $k=0,\ldots,N,$ $m=M+M_L+1,\ldots,M+M_L+M_U,$ and where we set $g(t,X_t)=0$ for $t>T$ and 
$Q\left(\left(t_k,X_{t_k}\right),0;\hat{\theta}\right)=0$ (no more exercises left).
The theoretical upper bound $U^n$ stemming from the dual problem in \cite{bender11} is given by:
\begin{multline*}
U^n=\EE\Bigg[\sup_{\substack{u_1,\ldots,u_n\in \NN\\ u_\nu\geq u_{\nu+1}+\delta}}\Bigg\{
\sum_{\nu=1}^{n-1}\Big(
g\left(u_\nu,X_{u_\nu}\right)-\left(
M^\nu(u_{\nu})-M^\nu(u_{\nu+1})\right)\\+A^{\nu}\left(u_{\nu+1}+\delta\right)-\ee{A^{\nu}\left(u_{\nu+1}+\delta\right)}[\mathcal{F}_{u_{\nu+1}}]
\Big)
+g\left(u_n,X_{u_n}\right)-M^n(u_n)
\Bigg\}\Bigg],    
\end{multline*}
We hence obtain $V_0^n\leq U^n$
and this bound is sharp for the exact Doob-Meyer decomposition terms $M^{*,\nu}$ and $A^{*,\nu},$
for $\nu=1,\ldots,n.$ We denote the sharp upper bound as
\begin{multline*}
U^{*,n}=\sup_{\substack{u_1,\ldots,u_n\in \NN\\ u_\nu\geq u_{\nu+1}+\delta}}\Bigg\{
\sum_{\nu=1}^{n-1}\Big(
g\left(u_\nu,X_{u_\nu}\right)-\left(
M^{*,\nu}(u_{\nu})-M^{*,\nu}(u_{\nu+1})\right)\\+A^{*,\nu}\left(u_{\nu+1}+\delta\right)-\ee{A^{*,\nu}\left(u_{\nu+1}+\delta\right)}[\mathcal{F}_{u_{\nu+1}}]
\Big)
+g\left(u_n,X_{u_n}\right)-M^{*,n}(u_n)
\Bigg\}.  
\end{multline*}
The following Monte Carlo average then yields an estimate of the upper bound for the optimal price $V_0^n:$
\begin{multline*}
\hat{U}^n=\frac{1}{M_U}\sum_{m=M+M_L+1}^{M+M_L+M_U}
\Bigg(\sup_{\substack{u_1,\ldots,u_n\in \NN\\ u_\nu\geq u_{\nu+1}+\delta}}
\sum_{\nu=1}^{n-1}\Big(
g\left(u_\nu,X^m_{u_\nu}\right)-
\left(M_m^\nu(u_\nu)-M_m^\nu(u_{\nu+1})\right)\\
+A_m^\nu\left(u_{\nu+1}+\delta\right)-\hat{\mathbb{E}}\left[A^\nu\left(u_{\nu+1}+\delta\right)\mid \mathcal{F}_{u_{\nu+1}}\right]
\Big)
+g\left(u_n,X^m_{u_n}\right)-
M_m^n\left(u_n\right)\Bigg).
\end{multline*} 
The pathwise supremum appearing in the expression of the upper bound can be computed using the recursion formula from Proposition 3.8 in \cite{bender15}. This recursion formula is implemented in our setting using the representation via the Q-function.

\subsection{Point estimate and confidence interval}
As in \cite{dos} and \cite{becker20}, we can construct a pointwise estimate for the optimal value in the multiple stopping framework in presence of a waiting time constraint by taking the pointwise estimate:
\begin{equation*}
\frac{\hat{L}^n+\hat{U}^n}{2}. 
\end{equation*}
By storing the empirical standard deviations for the lower and upper bounds that we denote as $\hat{\sigma}_{L^n}$ and $\hat{\sigma}_{U^n},$ respectively, one can leverage the central limit theorem as in Section \ref{sec:CI} to derive the asymptotic two-sided $(1-\alpha)$-confidence interval for the true optimal value $V_0^n:$
\begin{equation}
\left[\hat{L}^n-z_{\alpha/2}\frac{\hat{\sigma}_{L^n}}{\sqrt{M_L}},\hat{U}^n+z_{\alpha/2}\frac{\hat{\sigma}_{U^n}}{\sqrt{M_U}}\right].
\end{equation}
\subsection{Bias on the dual upper bound}
We now derive the extension of a result presented in \cite{meinshausen} on the bias resulting from the derivation of the upper bound, to the case of multiple stopping in presence of a waiting period. The dual problem from \cite{meinshausen}, being obtained from an optimization over a space of martingales and a set of stopping times, contains two terms: the bias coming from the martingale approximation, and the bias coming from the policy approximation. In the case with waiting constraint, as exemplified in the dual of \cite{bender11}, we show how one can again control the bias in the approximations to the $n$ Doob-Meyer decompositions of the Snell envelopes $y^{*,\nu},$ for $\nu=1,\ldots,n.$ Indeed, in the dual problem, each martingale $M^{*,\nu}$ is approximated by a martingale $M^\nu,$ and each predictable non-decreasing process $A^{*,\nu}$ is approximated by an integrable adapted process in discrete time $A^{\nu}.$
We proceed in three steps and analyse separately the bias from each approximation employed:
\begin{itemize}
    \item Martingale terms: \[\EE\Bigg[\sup_{\substack{u_1,\ldots,u_n\in \NN\\ u_\nu\geq u_{\nu+1}+\delta}}
\sum_{\nu=1}^{n-1}\Big\lvert
M^\nu(u_\nu)-M^\nu(u_{\nu+1})-\left(M^{*,\nu}(u_\nu)-M^{*,\nu}(u_{\nu+1})\right)\Big\lvert\Bigg]\]
    \item Adapted terms: \begin{multline*}
   \EE\Bigg[\sup_{\substack{u_1,\ldots,u_n\in \NN\\ u_\nu\geq u_{\nu+1}+\delta}}
\sum_{\nu=1}^{n-1}\Big\lvert A^{\nu}\left(u_{\nu+1}+\delta\right)-\ee{A^{\nu}\left(u_{\nu+1}+\delta\right)}[\mathcal{F}_{u_{\nu+1}}]\\-\left(
A^{*,\nu}\left(u_{\nu+1}+\delta\right)-\ee{A^{*,\nu}\left(u_{\nu+1}+\delta\right)}[\mathcal{F}_{u_{\nu+1}}]
\right)\Big\lvert\Bigg]   
    \end{multline*}
    \item Final term: \[\mathbb{E}\left[\sup_{0\leq n\leq N}\big\lvert g\left(u_n,X_{u_n}\right)-M^n(u_n)-
\left(g\left(u_n,X_{u_n}\right)-M^{*,n}(u_n)\right)\big\lvert\right]\]
\end{itemize}
The error in the final term $g\left(u_n,X_{u_n}\right)-M^n(u_n)$ can be bounded using the methodology in \cite{meinshausen}.
Define
\begin{equation*}
D_{y,n}=\sup_{\substack{\\0\leq k\leq N\\x \in E}}\Big\lvert y^{*,n}(t_k,x)-y^n(t_k,x)\Big\lvert, 
\end{equation*}
as the distance between the true Snell envelope and its approximation, and
\begin{equation*}
\sigma_{M_U^{inner},n}^2=\sup_{\substack{1\leq k\leq N\\x\in E}}\ee{\left(\hat{\mathbb{E}}\left[y^n(t_k,X_{t_k})\big\lvert X_{t_{k-1}}=x\right]-\EE\left[y^n(t_k,X_{t_k})\big\lvert X_{t_{k-1}}=x\right]\right)^2}[X_{t_{k-1}}=x], 
\end{equation*}
as an upper bound on the Monte Carlo error from the 1-step ahead nested simulation to approximate the continuation values.

In order to study the bias coming from the martingale approximations, we define
\begin{equation*}
D_{y}=\sup_{\substack{\nu=1,\ldots,n-1\\u_1,\ldots,u_n\in \NN\\u_\nu\geq u_{\nu+1}+\delta\\x \in E}}\Big\lvert y^{*,\nu}(u_\nu,x)-y^\nu(u_{\nu},x)\Big\lvert,  \end{equation*}
as the distance between the optimal Snell envelope and its approximation over all remaining exercise times,
\small
\begin{equation*}
\sigma_{M_U^{inner}}^2=   \sup_{\substack{\nu=1,\ldots,n-1\\u_1,\ldots,u_n\in \NN\\u_\nu\geq u_{\nu+1}+\delta\\x\in E}}\ee{\left(\hat{\mathbb{E}}\left[y^\nu(u_{\nu+1},X_{u_{\nu+1}})\big\lvert X_{u_{\nu}}=x\right]-\EE\left[y^\nu(u_{\nu+1},X_{u_{\nu+1}})\big\lvert X_{u_{\nu}}=x\right]\right)^2}[X_{u_{\nu}}=x],
\end{equation*}
\normalsize
and
\small
\begin{equation*}
\sigma_{M_U^{inner},\delta}^2=   \sup_{\substack{\nu=1,\ldots,n-1\\u_1,\ldots,u_n\in \NN\\u_\nu\geq u_{\nu+1}+\delta\\x\in E}}\ee{\left(\hat{\mathbb{E}}\left[y^\nu(u_\nu+\delta,X_{u_\nu+\delta})\big\lvert X_{u_{\nu}}=x\right]-\EE\left[y^\nu(u_\nu+\delta,X_{u_\nu+\delta})\big\lvert X_{u_{\nu}}=x\right]\right)^2}[X_{u_{\nu}}=x].
\end{equation*}
\normalsize
In other words, $\sigma_{M_U^{inner}}$ and $\sigma_{M_U^{inner},\delta}$ correspond to upper bounds on the standard deviations of the 1-step ahead and $\delta$-steps ahead Monte Carlo estimates of the continuation values, respectively, using a sample of $M_U^{inner}$ independent simulations starting from the endpoint of simulation path $m$ for $m=M+M_L+1,\ldots,M+M_L+M_U$.

The following theorem allows to control for the bias in the derivation of the upper bound from the dual problem.
\begin{theorem}[Dual upper bound bias]\label{thm:thm2}
Let $B_\delta^n\left(M,A\right)$ be the total bias which is the difference between the approximate upper bound using $\left(M^{\nu}\right)_{\nu=1,\ldots,n}$ and $\left(A^{\nu}\right)_{\nu=1,\ldots,n-1}$ and the theoretical sharp upper bound using the optimal Doob decomposition components $\left(M^{*,\nu}\right)_{\nu=1,\ldots,n}$ and $\left(A^{*,\nu}\right)_{\nu=1,\ldots,n-1}$
\[
B_\delta^n\left(M,A\right)=U^n-U^{*,n}.
\]
The following result holds:
\small
\begin{equation}
B_\delta^n\left(M,A\right)\leq  
8(n-1)\sqrt{\left(4D_y^2+\sigma_{M_U^{inner}}^2\right)T}+(n-1)\left(\sigma_{M_U^{inner},\delta}+D_y\right)+2\sqrt{\left(4D_{y,n}^2+\sigma_{M_U^{inner},n}^2\right)T}.
\end{equation}
\end{theorem}
\normalsize

\noindent
In order to prove this result, let us state an intermediary result which will appear in the proofs of the following propositions.
Define 
\begin{equation*}
R_t^\nu=M_t^\nu-M_t^{*,\nu},    
\end{equation*}
as the difference between the martingale approximation and the optimal martingale for the problem with $\nu$ remaining exercise times, for $\nu=1,\ldots,n.$
\begin{lemma}\label{lem:lem1}
The process $R^\nu$ is a martingale with $R^\nu(0)=0,$ for all $\nu=1,\ldots,n,$ and we have the following inequality on the second moment of the martingale increments, for all $0\leq t<T$ and $\nu=1,\ldots,n:$
\begin{equation*}
\ee{\left(R_{t+1}^\nu-R_{t}^\nu\right)^2}[\mathcal{F}_t]\leq 4D_y^2+\sigma_{M_U^{inner}}^2.    
\end{equation*}
As a consequence,
\begin{equation*}
\ee{\left(R_{t}^\nu\right)^2}\leq\left(4D_y^2+\sigma_{M_U^{inner}}^2\right)t.    
\end{equation*}
\begin{proof}
The proof of this lemma follows similar lines to the proof of Lemma 6.1 in \cite{meinshausen}. Let $\nu\in\{1,\ldots,n\}.$ As a difference of martingales with initial value 0, $R^\nu$ is also a martingale with initial value 0.
The increments can be rewritten as
\small
\begin{align*}
R_{t+1}^\nu-R_{t}^\nu&=
M_{t+1}^\nu-M_{t+1}^{*,\nu}-\left(M_t^\nu-M_t^{*,\nu}\right)\\
&=y^\nu\left(t+1,X_{t+1}\right)-\hat{\EE}\left[y^\nu\left(t+1,X_{t+1}\right)\mid\mathcal{F}_t\right]-
\left(y^{*,\nu}\left(t+1,X_{t+1}\right)-\ee{y^{*,\nu}\left(t+1,X_{t+1}\right)}[\mathcal{F}_t]\right)\\
&=y^{\nu}\left(t+1,X_{t+1}\right)-y^{*,\nu}\left(t+1,X_{t+1}\right)
+\ee{y^{*,\nu}\left(t+1,X_{t+1}\right)}[\mathcal{F}_t]-\ee{y^{\nu}\left(t+1,X_{t+1}\right)}[\mathcal{F}_t]\\
&+\ee{y^{\nu}\left(t+1,X_{t+1}\right)}[\mathcal{F}_t]-\hat{\EE}\left[y^\nu\left(t+1,X_{t+1}\right)\mid\mathcal{F}_t\right].
\end{align*}
\normalsize
Now, both differences between the first two terms and the third and fourth term in the final equality are bounded in absolute value by $D_y.$ The last term corresponds to the error from the Monte Carlo approximation of the 1-step ahead continuation values. Since this error term has mean 0, a second moment bounded by $\sigma_{M^{inner}_U}^2$ and is independent of the term
$\left(y^{*,\nu}\left(t+1,X_{t+1}\right)-y^{\nu}\left(t+1,X_{t+1}\right)\right),$ we obtain the desired result.
\end{proof}
\end{lemma}

\begin{proposition}\label{prop:biasFinal}
The bias in the final term can be bounded by
\begin{equation}
\mathbb{E}\left[\sup_{0\leq n\leq N}\Big\lvert g\left(u_n,X_{u_n}\right)-M^n(u_n)-
\left(g\left(u_n,X_{u_n}\right)-M^{*,n}(u_n)\right)\Big\lvert\right]\leq 2\sqrt{\left(4D_{y,n}^2+\sigma_{M_U^{inner},n}^2\right)T}.
\end{equation}
\end{proposition}
\begin{proof}
We consider the error in the final term
\begin{equation*}
\EE\bigg[\sup_{\substack{0\leq n\leq N}}\Big\lvert M^n(u_n)-M^{*,n}(u_n)\Big\lvert\bigg]
\leq \EE\bigg[
\sup_{\substack{0\leq t\leq T}}\big\lvert R_t^n\big\lvert
\bigg].
\end{equation*}
From the Cauchy-Schwarz inequality,
\begin{equation*}
\EE\Big[
\sup_{\substack{0\leq t\leq T}}\big\lvert R_t^n\big\lvert
\Big]\leq 
\Big(\EE\Big[
\sup_{\substack{0\leq t\leq T}} \left(R_t^n\right)^2
\Big]\Big)^{1/2},
\end{equation*}
and since $R^n$ is a martingale, $(R^n)^2$ is a non-negative submartingale which is well-defined from the existence of $D_y$ and $\sigma_{M_U^{inner}}.$
Then, using Doob's submartingale inequality,
\begin{equation*}
\EE\Big[
\sup_{\substack{0\leq t\leq T}}\left(R_t^n\right)^2
\Big]\leq 4\EE\big[\left(R_T^n\right)^2\big].
\end{equation*}
This last inequality in combination with Lemma \ref{lem:lem1} leads to the desired result.
\end{proof}

\begin{proposition}\label{prop:biasMart}
The bias from the approximations of the martingale terms can be bounded by
\begin{multline}
\EE\Bigg[\sup_{\substack{u_1,\ldots,u_n\in \NN\\ u_\nu\geq u_{\nu+1}+\delta}}
\sum_{\nu=1}^{n-1}\Big\lvert
M^\nu(u_\nu)-M^\nu(u_{\nu+1})-\left(M^{*,\nu}(u_\nu)-M^{*,\nu}(u_{\nu+1})\right)\Big\lvert\Bigg]\\
\leq 4(n-1)\sqrt{\left(4D_y^2+\sigma_{M_U^{inner}}^2\right)T}.  
\end{multline}
\end{proposition}
\begin{proof}
The error in the martingale term for the problem with $\nu$ remaining exercise times can be expressed as
\sloppy
\begin{align*}
\big\lvert
M^\nu(u_\nu)-M^\nu(u_{\nu+1})-\left(M^{*,\nu}(u_\nu)-M^{*,\nu}(u_{\nu+1})\right)\big\lvert
&\leq 
\big\lvert
R^\nu(u_\nu)\big\lvert+
\big\lvert R^\nu(u_{\nu+1})\big\lvert\\
&\leq \sup_{\substack{u_1,\ldots,u_n\in \NN\\ u_\nu\geq u_{\nu+1}+\delta}}\big\lvert
R^\nu(u_\nu)\big\lvert+
\sup_{\substack{u_1,\ldots,u_n\in \NN\\ u_\nu\geq u_{\nu+1}+\delta}}
\big\lvert R^\nu(u_{\nu+1})\big\lvert.
\end{align*}
By taking the sum over $\nu=1,\ldots,n-1,$ taking the supremum over the subspace of $\NN^n$ with the constraints imposed by the presence of the waiting period, and finally taking the expectation, we obtain
\begin{multline*}
\EE\Bigg[\sup_{\substack{u_1,\ldots,u_n\in \NN\\ u_\nu\geq u_{\nu+1}+\delta}}
\sum_{\nu=1}^{n-1}\Big\lvert
M^\nu(u_\nu)-M^\nu(u_{\nu+1})-\left(M^{*,\nu}(u_\nu)-M^{*,\nu}(u_{\nu+1})\right)\Big\lvert\Bigg]\\
\leq (n-1)
\EE\Bigg[\sup_{\substack{u_1,\ldots,u_n\in \NN\\ u_\nu\geq u_{\nu+1}+\delta}}\big\lvert
R^\nu(u_\nu)\big\lvert\Bigg]+(n-1)
\EE\Bigg[\sup_{\substack{u_1,\ldots,u_n\in \NN\\ u_\nu\geq u_{\nu+1}+\delta}}\big\lvert
R^\nu(u_{\nu+1})\big\lvert\Bigg].
\end{multline*}
Following Proposition \ref{prop:biasFinal}, using first the Cauchy-Schwarz inequality and then Doob's submartingale inequality in both terms of the right-hand side of the previous inequality, we obtain the desired result.
\end{proof}

\noindent
Finally, the bias coming from the approximations of the non-decreasing predictable processes can be controlled as stated in the following proposition.
\begin{proposition}\label{prop:biasBV}
The bias from the approximations of the non-decreasing predictable terms can be bounded by
\small
\begin{multline*}
\EE\Bigg[\sup_{\substack{u_1,\ldots,u_n\in \NN\\ u_\nu\geq u_{\nu+1}+\delta}}
\sum_{\nu=1}^{n-1}\Big\lvert A^{\nu}\left(u_{\nu+1}+\delta\right)-\ee{A^{\nu}\left(u_{\nu+1}+\delta\right)}[\mathcal{F}_{u_{\nu+1}}]\\-\left(
A^{*,\nu}\left(u_{\nu+1}+\delta\right)-\ee{A^{*,\nu}\left(u_{\nu+1}+\delta\right)}[\mathcal{F}_{u_{\nu+1}}]
\right)\Big\lvert\Bigg]\leq 
(n-1)\left(\sigma_{M_U^{inner},\delta}+D_y+4\sqrt{\left(4D_y^2+\sigma_{M_U^{inner}}^2\right)T}\right).
\end{multline*}
\end{proposition}
\normalsize
\begin{proof}
Again, we consider the approximation of the predictable process for the problem with $\nu$ remaining exercise times
\small
\begin{multline*}
A^{\nu}\left(u_{\nu+1}+\delta\right)-\ee{A^{\nu}\left(u_{\nu+1}+\delta\right)}[\mathcal{F}_{u_{\nu+1}}]-\left(
A^{*,\nu}\left(u_{\nu+1}+\delta\right)-\ee{A^{*,\nu}\left(u_{\nu+1}+\delta\right)}[\mathcal{F}_{u_{\nu+1}}]
\right)\\
=M^\nu\left(u_{\nu+1}+\delta\right)-M^\nu(u_{\nu+1})+\ee{
y^\nu\left(u_{\nu+1}+\delta,X_{u_{\nu+1}+\delta}\right)}[\mathcal{F}_{u_{\nu+1}}]-y^\nu\left(u_{\nu+1}+\delta, X_{u_{\nu+1}+\delta}\right)\\
-\left(
M^{*,\nu}\left(u_{\nu+1}+\delta\right)-M^{*,\nu}(u_{\nu+1})+\ee{
y^{*,\nu}\left(u_{\nu+1}+\delta, X_{u_{\nu+1}+\delta}\right)}[\mathcal{F}_{u_{\nu+1}}]-y^{*,\nu}\left(u_{\nu+1}+\delta,X_{u_{\nu+1}+\delta}\right)
\right).
\end{multline*}
\normalsize
Now, since
\begin{equation*}
\big\lvert M^\nu\left(u_{\nu+1}+\delta\right)-M^{*,\nu}\left(u_{\nu+1}+\delta\right) \big\lvert\leq \big\lvert R^\nu_{u_{\nu+1}+\delta}\big\lvert,
\end{equation*}
\begin{equation*}
\big\lvert M^\nu\left(u_{\nu+1}\right)-M^{*,\nu}\left(u_{\nu+1}\right) \big\lvert\leq \big\lvert R^\nu_{u_{\nu+1}}\big\lvert,
\end{equation*}
and
\begin{equation*}
\big\lvert y^{*,\nu}\left(u_{\nu+1}+\delta,X_{u_{\nu+1}+\delta}\right)- y^{\nu}\left(u_{\nu+1}+\delta,X_{u_{\nu+1}+\delta}\right)\big\lvert \leq D_y,
\end{equation*}
by summing over all exercise opportunities, taking the supremum and then the expectation, we obtain by definition of $\sigma_{M_U^{inner},\delta},$
\small
\begin{multline*}
\EE\bigg[\sup_{\substack{u_1,\ldots,u_n\in \NN\\ u_\nu\geq u_{\nu+1}+\delta}}
\sum_{\nu=1}^{n-1}\Big\lvert A^{\nu}\left(u_{\nu+1}+\delta\right)-\ee{A^{\nu}\left(u_{\nu+1}+\delta\right)}[\mathcal{F}_{u_{\nu+1}}]\\-\left(
A^{*,\nu}\left(u_{\nu+1}+\delta\right)-\ee{A^{*,\nu}\left(u_{\nu+1}+\delta\right)}[\mathcal{F}_{u_{\nu+1}}]
\right)\Big\lvert\bigg]\leq (n-1)\left(\sigma_{M_U^{inner},\delta}+D_y+4\sqrt{\left(4D_y^2+\sigma_{M_U^{inner}}^2\right)T}\right).
\end{multline*}
\end{proof}
\normalsize

\noindent
The proof of Theorem \ref{thm:thm2} is then obtained by summing up all contributions to the total bias from Propositions \ref{prop:biasFinal}, \ref{prop:biasMart} and \ref{prop:biasBV}. We thus obtain an upper bound on the total bias stemming from the errors in all approximations.
We see in particular in the expression of the total bias that the waiting period appears implicitly in the error term from the Monte Carlo $\delta$-steps ahead estimation.

We now illustrate the Q-learning approach with several numerical examples.

\section{Numerical results}\label{sec:results}
As illustrative examples we present 
swing options in the multiple stopping framework in several dimensions, with varying maturities, $n=2$ exercise rights and a waiting period constraint $\delta>0$.

In all examples we select mini-batches of size 1000 using experience replay on a sample of 1,000,000 simulations. We consider ReLU activation functions applied component-wise, perform stochastic gradient descent for the optimization step using the RMSProp implementation from \texttt{PyTorch}, and initialize the network parameters using the default \texttt{PyTorch} implementation.

Swing options appear in the commodity and energy markets (natural gas, electricity) as hedging instruments to protect investors from futures price fluctuations. They give the holder of the option the right to exercise at multiple times during the lifetime of the contract, the number of exercise opportunities being specified at inception. Further constraints can be imposed at each exercise time, such as the maximal quantity of energy that can be bought or sold, or the minimal waiting period between two exercise times, see e.g. \cite{bender11}. In the presence of a volume constraint, under certain sufficient conditions, see \cite{bardou},  the optimal policy is a so-called "bang-bang strategy", see e.g. \cite{daluiso20}, i.e. at each exercise time the optimal strategy is to buy or sell the maximum or the minimum amount allowed, which then simplifies the action space. A model for commodity futures prices is derived in \cite{daluiso20}, implemented using proximal policy optimization (PPO), which is another tool from reinforcement learning and where the policy update is forced to be close to the previous policy by clipping the advantage function. The pricing of such contracts is also investigated in \cite{meinshausen} with no constraints, in \cite{bender11} with waiting time constraint and in \cite{bender15} with both waiting time and volume constraints.
We will consider the same model for the electricity spot prices as in \cite{meinshausen}, that is, the exponential of a Gaussian Ornstein-Uhlenbeck process, which in discrete time takes the form
\begin{equation*}
\log S_{t+1}=\left(1-k\right)\left(\log S_{t}-\mu\right)+\mu+\sigma Z_t,
\end{equation*}
where $\{Z_t\}_{t=0,\ldots,T-1}$ are standard normal random variables, and where we choose $\sigma=0.5,$ $k=0.9$, $\mu=0,$ $S_0=1$ and strike price $K=1.$ We consider the payoff $(S_t-K)^{+}$ for time $t=0,\ldots,T,$ without any discounting, as in \cite{meinshausen}, \cite{bender11} and \cite{bender15}. A discount factor could be taken into account with no real additional complexity. 
In the multi-dimensional setting we will consider the same payoff as max-call options, that is $\left(\max_{i=1,\ldots,d}S_t^i-K\right)^{+}$ for a $d$-dimensional vector of asset prices $\left(S^1,\ldots,S^d\right)^T,$ where we assume for the marginals the same dynamics as above and independence between the respective innovations. We will consider the same starting value $S_0=1$ for all the assets in the examples below.
We stress that this pricing approach can be extended to any other type of Markovian dynamics which are more adequate for capturing electricity prices. 

We assume that the arbitrage-free  price is given by taking the expectation at (\ref{eq:stopMult}) under an appropriate pricing measure, that is, a probability measure under which the (discounted) prices of tradable and storable basic securities in the underlying market are (local) martingales. The electricity market being incomplete, the prices will depend on the choice of the pricing measure. The latter can be selected by considering a calibration on liquidly traded swing options.

We select a deep neural network with 3 hidden layers containing 32 neurons each for the examples with $d=3$ and $d=10,$ and 90 neurons each for the examples with $d=50.$
We present our results in dimensions $d=3,$ $d=10$ and $d=50$ in Table \ref{tab:swing3d} below, using $M_L=100,000$, $M_U=100$ and $J=5000.$
\begin{table}[ht!]
	\caption{Prices at $t=0$ for swing options with varying maturities and asset price dimensions, $K=1,$ $\mu=0,$ $\sigma=0.5,$ and $k=0.9.$}\label{tab:swing3d}
	\begin{center}
		\begin{tabular}{|c|c|c|c|c|}
			\hline
			Model parameters  & $\hat{L}$ & PE & $\hat{U}$ & CI \\ \hline
			$d=3, n=2, \delta=2, T=10$ & 2.7249 & 2.8269 & 2.9288 & [2.7181,\mbox{ }3.0319] \\ \hline
			$d=3, n=2, \delta=2, T=20$  & 3.4934 & 3.8283 & 4.1632 & [3.4864,\mbox{ }4.3362] \\ \hline
			$d=10, n=2, \delta=2, T=10$ & 4.1343 & 4.3312 & 4.5281 & [4.1268,\mbox{ }4.6886]\\ \hline
			$d=10, n=2, \delta=2, T=20$ & 4.9706 & 5.5155 & 6.0604 & [4.9629,\mbox{ }6.1922]\\ \hline
			$d=50, n=2, \delta=2, T=10$ & 6.2141 & 6.6333 & 7.0525 & [6.2058,\mbox{ }7.2704]\\ \hline
			$d=50, n=2, \delta=2, T=20$ & 7.0785 & 7.9207 & 8.7628 & [7.0702,\mbox{ }9.0418]\\ \hline
		\end{tabular}
	\end{center}
\end{table}

\section{Conclusion}\label{sec:conclu}
We have presented optimal stopping problems appearing in the valuation of financial products under the lens of reinforcement learning. This new angle allows us to model the optimal action-value function using the RL machinery and deep neural networks. This method could serve as an alternative to recent approaches developed in the literature, be it to derive the optimal policy by modeling directly the stopping times as in \cite{dos}, or by modeling the continuation values by approximating conditional expectations as in \cite{becker20}. We have also considered the pricing of multiple exercise stopping problems with waiting period constraint and derived lower and upper bounds on the option price, using the trained neural network and the dual representation, respectively. In addition, we have proved a result that controls for the total bias resulting from the approximation of the terms appearing in the dual formulation. The RL framework is suitable for configurations where the action space varies in a non-trivial way with time, i.e. there are certain degrees of freedom for the agent to explore the environment at each time step. This is exemplified through the swing option with multiple stopping rights and waiting time constraint, but could also be useful for more complex environments.
It could also be interesting to investigate state-of-the-art improvements to the DQN algorithm brought forward in \cite{rainbow}. One could explore these avenues in further research. 

\newpage
\section*{Acknowledgements}{We thank Prof. Patrick Cheridito for helpful comments and for carefully reading previous versions of the manuscript.
\\As SCOR Fellow, John Ery thanks SCOR for financial support.
\\Both authors have contributed equally to this work.}

\bibliographystyle{apa}
\bibliography{bibli}

\end{document}